\documentclass[conference]{IEEEtran}
\usepackage{amssymb}
\usepackage{amsmath}
\usepackage{amsfonts}
\usepackage{float}
\usepackage{graphics}
\usepackage[tight,footnotesize]{subfigure}
\usepackage{rotating}
\usepackage{algorithm}
\usepackage{algorithmic}
\usepackage{enumerate}
\usepackage{cite}

\newcommand*{\QEDA}
{\hfill\ensuremath{\blacksquare}}
\DeclareMathSizes{10}{9}{7}{5}
\newtheorem{theo}{Theorem}

\newtheorem{theorem}{Theorem}

\newtheorem{definition}[theorem]{Definition}

\newtheorem{lemma}{Lemma}

\newtheorem{remark}{Remark}

\input{tcilatex}
\begin{document}

\title{Unshared Secret Key Cryptography: Finite Constellation Inputs and
Ideal Secrecy Outage} \pubid{} \specialpapernotice{}
\author{\authorblockN{Shuiyin  Liu, Yi Hong, and Emanuele Viterbo\thanks
{This work is supported by ARC under Grant Discovery Project No.
DP130100336.}}
\authorblockA{ECSE Department, Monash University\\
Melbourne, VIC 3800, Australia\\
Email: shuiyin.liu, yi.hong, emanuele.viterbo@monash.edu}}
\maketitle

\begin{abstract}
The \emph{Unshared Secret Key Cryptography} (USK), recently proposed by the
authors, guarantees Shannon's \emph{ideal secrecy} and \emph{perfect secrecy}
for MIMO wiretap channels, without requiring secret key exchange. However,
the requirement of infinite constellation inputs limits its applicability to
practical systems. In this paper, we propose a practical USK scheme using
finite constellation inputs. The new scheme is based on a \emph{cooperative
jamming} technique, and is valid for the case where the eavesdropper has
more antennas than the transmitter. We show that Shannon's ideal secrecy can
be achieved with an arbitrarily small outage probability.
\end{abstract}

\section{Introduction}

Symmetric-key cryptography (e.g. AES \cite{Daemen:2002:DRA}) has
traditionally been the major technology for providing a secure gateway for
communication and data exchanges at the network layer. One weakness of this
approach is that the transmitter (Alice) and the legitimate receiver (Bob)
must trust some secure communications channel to transmit the secret key,
and there may be a chance that others (Eve) can discover the secret key
during transmission. Physical layer security (PLS) is an alternative way of
providing non-key based security solutions \cite{Wyner75}. The PLS
approaches leverage state-of-the art channel coding (e.g. polar codes \cite%
{Mahdavifar11}) to enhance security at the physical layer. The general
problem of PLS is the requirement of an infinite-length wiretap code to
approach the \emph{secrecy capacity}\cite{Hellman78}. This limits the
applicability of these schemes to practical communication systems.

Our previous work \cite{Liu13Letter,Shuiyin13_3,Liu147} has shown
that it is possible to protect the confidential message without
requiring either secret key exchange or wiretap codes. In
particular, we proposed the \emph{Unshared
Secret Key Cryptography} (USK) to comply with two security goals: \textit{(i)%
} the secret key is not needed by Bob to decipher, \textit{(ii)} the secret
key is fully affecting Eve's ability to decipher the ciphertext. Although
those two goals seem to contradict each other, this can be reconciled by
aligning a one-time pad (OTP) secret key within the null space of a MIMO
channel between Alice and Bob. In this way, the OTP nulls out at Bob, but
adds a certain degree of uncertainty to the received signal at Eve. The USK
is rooted in the \emph{artificial noise} (AN) technique \cite{Goel08}, and
has been shown to achieve Shannon's \emph{ideal secrecy} and \emph{perfect
secrecy }\cite{Liu147}. An interesting case is ideal secrecy: an encryption
algorithm is ideally secure if no matter how much of ciphertext is
intercepted by Eve, there is no unique solution of the plaintext but many
solutions of comparable probability \cite{Shannon46}. An ideal cryptosystem
has \emph{information-theoretic security} (i.e., cryptanalytically
unbreakable), but not Shannon's perfect secrecy (i.e., plaintext and
ciphertext are mutually independent) \cite{Shannon46}.

The original USK scheme \cite{Liu147} may be regarded as being of
theoretical interest only, since it bases on two assumptions:\ \textit{(i)}
infinite lattice constellation input, \textit{(ii)} Alice has more antennas
than Eve. The first assumption is used to prove the ideal secrecy, and the
second assumption ensures that Eve cannot run zero-forcing (ZF) attack to
remove the OTP \cite{Shuiyin14_0}. In this work, we remove these assumptions
and show how the idea of USK can be applied to practical systems. We put
forward a new security model and measure:

\begin{enumerate}
\item \emph{Finite constellation inputs}: we use finite input alphabets
based on \textit{QAM} signalling.

\item \emph{Cooperative jammers}: the OTP is generated from the third-party
jammers. This renders the USK viable for the cases where Eve has more
antennas than Alice.

\item \emph{Ideal secrecy outage}: we show that Shannon's ideal secrecy can
be achieved for finite constellation input with an arbitrarily small outage
probability.
\end{enumerate}

Section II presents the system model. Section III describes the USK
cryptosystem with finite constellation input.\ Section IV analyzes the
security of the USK cryptosystem. Section V sets out the theoretical and
practical conclusions. The Appendix contains the proofs of the theorems.

\textit{Notation:} Matrices and column vectors are denoted by upper and
lowercase boldface letters, and the Hermitian transpose, inverse,
pseudo-inverse of a matrix $\mathbf{B}$ by $\mathbf{B}^{H}$, $\mathbf{B}%
^{-1} $, and $\mathbf{B}^{\dagger }$, respectively. We use the standard
asymptotic notation $f\left( x\right) =O\left( g\left( x\right) \right) $
when $\lim \sup\limits_{x\rightarrow \infty }|f(x)/g(x)|<\infty $. The real,
complex, integer, and complex integer numbers are denoted by $\mathbb{R}$, $%
\mathbb{C} $, $\mathbb{Z}$, and $\mathbb{Z}\left[ i\right] $, respectively. $%
H(\cdot )$, $H(\cdot |\cdot )$, and $I(\cdot ; \cdot )$ represent entropy,
conditional entropy, and mutual information, respectively.

\section{System Model}

\subsection{Cooperative Jamming}

The security model is based on our recently proposed MIMO cooperative
jamming scheme using artificial noise \cite{Shuiyin14_isita}. This model is
quite different from the conventional cooperative jamming scheme (multiuser
jammers over AWGN channels) in \cite{Tekin08}. The detailed differences can
be referred to \cite{Shuiyin14_isita}.

In our setting, we consider a MIMO wiretap system model that includes a
transmitter (Alice), a legitimate receiver (Bob), and a passive eavesdropper
(Eve), with $N_{\text{A}}$, $N_{\text{B}}$, and $N_{\text{E}}$ antennas,
respectively. Additionally, a set of $N$ friendly jammers $\left\{ \text{J}%
_{i}\right\} _{1}^{N}$ are used to protect against eavesdropping, where each
one has $N_{\text{J,}i}$ antennas, respectively. We assume that $N_{\text{B}%
}\geq N_{\text{A}}$ and $N_{\text{J,}i}>N_{\text{B}}$. Let
\begin{equation}
\setlength{\abovedisplayskip}{3pt} \setlength{\belowdisplayskip}{3pt} N_{%
\text{J}}=\sum_{i=1}^{N}N_{\text{J,}i}\text{,}
\end{equation}
be the total number of antennas among all the jammers.

Alice sends a information vector $\mathbf{u}$, which is uniformly selected
from a $M$-QAM constellation $\mathcal{Q}^{N_{\text{A}}}$, where $\Re (%
\mathcal{Q}\mathbf{)}=\Im (\mathcal{Q}\mathbf{)}=\{0,1,...,\sqrt{M}-1\}$.
For simplicity, we ignore the shifting and scaling operations at Alice to
minimize the transmit power.

Let the matrices $\left\{ \mathbf{H}_{\text{JB,}i}\in \mathbb{C}
^{N_{\text{B}}\times N_{\text{J,}i}}\right\} _{1}^{N}$ represent the
channels from $\text{J}_{i}$ to Bob, for $1\leq i\leq N$. Suppose
that $\text{J}_{i}$ knows $\mathbf{H}_{\text{JB,}i}$, using the AN
scheme \cite{Goel08}, the $i^{\text{th}}$ jammer transmits%
\begin{equation}
\setlength{\abovedisplayskip}{3pt} \setlength{\belowdisplayskip}{3pt}
\mathbf{x}_{\text{J},i}=\mathbf{Z}_{i}\mathbf{v}_{i}\text{,}
\end{equation}%
where $\mathbf{Z}_{i}=\mbox{null}(\mathbf{H}_{\text{JB,}i})\in
\mathbb{C} ^{N_{\text{J,}i}\times (N_{\text{J,}i}-N_{\text{B}})}$
represents an
orthonormal basis of the null space of $\mathbf{H}_{\text{JB,}i}$, i.e., $%
\mathbf{H}_{\text{JB,}i}\mathbf{Z}_{i}=\mathbf{0}_{N_{\text{B}}\times (N_{%
\text{J,}i}-N_{\text{B}})}$. Each J$_{i}$ randomly and independently
(without any predefined distribution) chooses a vector
$\mathbf{v}_{i}\in \mathbb{C}^{N_{\text{J,}i}-N_{\text{B}}}$.

We set a peak jamming power constraint $P_{\text{J}}$, i.e.,%
\begin{equation}
\setlength{\abovedisplayskip}{3pt} \setlength{\belowdisplayskip}{3pt} P_{%
\text{J}}\geq \sum_{i=1}^{N}||\mathbf{x}_{\text{J},i}||^{2}=\sum_{i=1}^{N}||%
\mathbf{v}_{i}||^{2}\text{.}
\end{equation}

The signals received by Bob and Eve are given by%
\begin{align}
\mathbf{z}& =\mathbf{Hu}+\sum_{i=1}^{N}\mathbf{H}_{\text{JB,}i}\mathbf{x}_{%
\text{J},i}+\mathbf{n}_{\text{B}}=\mathbf{Hu}+\mathbf{n}_{\text{B}}\text{,}
\label{z} \\
\mathbf{y}& =\mathbf{Gu}+\sum_{i=1}^{N}\mathbf{H}_{\text{JE,}i}\mathbf{x}_{%
\text{J},i}+\mathbf{n}_{\text{E}}=\mathbf{Gu}+\mathbf{\hat{H}}_{\text{JE}}%
\mathbf{\hat{Z}\hat{v}}+\mathbf{n}_{\text{E}}\text{,}  \label{y}
\end{align}%
where $\mathbf{\hat{H}}_{\text{JE}}=[\mathbf{H}_{\text{JE,}1}$, ... , $\mathbf{H}%
_{\text{JE,}N}]$, $\mathbf{\hat{Z}}=\mathrm{diag}(\left\{ \mathbf{Z}%
_{i}\right\} _{1}^{N})$, and $\mathbf{\hat{v}}=[\mathbf{v}_{1}^{T}$, ... , $%
\mathbf{v}_{N}^{T}]^{T}$.

The matrices $\mathbf{H}\in \mathbb{C}^{N_{\text{B}}\times
N_{\text{A}}},\mathbf{G}\in \mathbb{C}^{N_{\text{E}}\times
N_{\text{A}}}$ represent the channel from Alice to Bob
and Alice to Eve, respectively, while $\mathbf{H}_{\text{JE,}i}\in \mathbb{C}%
^{N_{\text{E}}\times N_{\text{J,}i}}$ is the channel from jammer J$_{i}$ to
Eve. All channel matrices are assumed to be mutually independent (i.e., all
terminals are not co-located) and have i.i.d. entries $\sim \mathcal{N}_{%
\mathbb{C}}(0$, $1)$. Except for Eve, no one knows $\mathbf{G}$ and $\mathbf{%
\hat{H}}_{\text{JE}}$. We assume that $\mathbf{n}_{\text{B}}$ and $\mathbf{n}%
_{\text{E}}$ are white Gaussian noise (AWGN) vectors at Bob and Eve,
respectively, with i.i.d. entries $\sim \mathcal{N}_{\mathbb{C}}(0$, $\sigma
_{\text{B}}^{2})$ and $\mathcal{N}_{\mathbb{C}}(0$, $\sigma _{\text{E}}^{2})$%
.

\begin{remark}
The vector $\mathbf{\hat{v}}$ is independently generated by jammers, but not
needed by Bob to decipher, while it is fully affecting Eve's ability to
decipher the original message $\mathbf{u}$. This enables us to interpret $%
\mathbf{\hat{v}}$ as an \emph{unshared} OTP.
\end{remark}

\subsection{Channel Assumptions}

We consider the worst-case scenario for Alice and Bob:

\begin{itemize}
\item Alice does not know any channel matrix.

\item Bob only knows $\mathbf{H}$.

\item The $i^{\text{th }}$jammer $\text{J}_{i}$ only knows $\mathbf{H}_{%
\text{JB,}i}$, for all $i$.

\item Eve has perfect knowledge of all channel matrices.

\item No upper bound on $N_{\text{E}}$ or Eve's SNR.
\end{itemize}

We focus on information-theoretic security, hence, our analysis will
concentrate on Eve's equivocation $H(\mathbf{{{\mathbf{u}}}|y})$. For
simplicity, we further assume that Eve's channel is noiseless:%
\begin{equation}
\setlength{\abovedisplayskip}{3pt} \setlength{\belowdisplayskip}{3pt}
\mathbf{y}=\mathbf{Gu}+\mathbf{\hat{H}}_{\text{JE}}\mathbf{\hat{Z}\hat{v}}%
\text{.}  \label{WC_S}
\end{equation}%
Using \emph{Data Processing Inequality}, it is simple to show that Eve's
channel noise can only increase her equivocation:%
\begin{equation*}
H(\mathbf{{{\mathbf{u}}}|Gu}+\mathbf{\hat{H}}_{\text{JE}}\mathbf{\hat{Z}\hat{%
v}})\leq H(\mathbf{{{\mathbf{u}}}|Gu}+\mathbf{\hat{H}}_{\text{JE}}\mathbf{%
\hat{Z}\hat{v}+n}_{\text{E}})\text{.}
\end{equation*}

Moreover, we assume that%
\begin{equation}
\setlength{\abovedisplayskip}{3pt} \setlength{\belowdisplayskip}{3pt} N_{%
\text{E}}<\sum_{i=1}^{N}N_{\text{J,}i}=N_{\text{J}}\text{,}  \label{con_NJ}
\end{equation}%
That requirement ensures that Eve cannot remove the OTP $\mathbf{\hat{v}}$
by running ZF equalization. To see this, we first let%
\begin{equation}
\mathbf{\hat{H}}_{\text{JB}}=[\mathbf{H}_{\text{JB,}1},\text{... },\mathbf{H}%
_{\text{JB,}N}]\text{.}
\end{equation}%
If $N_{\text{E}}\geq N_{\text{J}}$, $\mathbf{\hat{H}}_{\text{JE}}$ has a
left inverse, denoted by $\mathbf{\hat{H}}_{\text{JE}}^{\dag }$, then the
term $\mathbf{\hat{H}}_{\text{JE}}\mathbf{\hat{Z}\hat{v}}$ in (\ref{WC_S})
can be removed by multiplying $\mathbf{y}$ by $\mathbf{{{\mathbf{W}}}=%
\mathbf{\hat{H}}}_{\text{JB}}\mathbf{\hat{H}}_{\text{JE}}^{\dag }$, i.e.,%
\begin{equation}
\mathbf{\mathbf{W}\mathbf{y}}=\mathbf{{{\mathbf{WGu+}}}\mathbf{\hat{H}}}_{%
\text{JB}}\mathbf{\hat{Z}\hat{v}={{\mathbf{WGu}}}}\text{.}
\end{equation}%
Note that if $N_{\text{E}}<N_{\text{J}}$, this operation is not possible.

\begin{remark}
The cooperative jamming approach allows us to replace the constraint $N_{%
\text{E}}<N_{\text{A}}$ in the original USK scheme \cite{Liu147} by $N_{%
\text{E}}<N_{\text{J}}$, which is much easier to accomplish by increasing
the number of jammers.
\end{remark}

\subsection{Shannon's Ideal Secrecy}

We consider a cryptosystem where a sequence of $K$ messages $\left\{ \mathbf{%
m}_{i}\right\} _{1}^{K}$ are enciphered into the cryptograms $\left\{
\mathbf{y}_{i}\right\} _{1}^{K}$ using a sequence of secret keys $\left\{
k_{i}\right\} _{1}^{K}$. We assume that $\left\{ \mathbf{m}_{i}\right\}
_{1}^{K}$ and $\left\{ k_{i}\right\} _{1}^{K}$ are mutually independent. Let
$L_{i}$ be the space size of $k_{i}$. Shannon introduced the concept of
ideal secrecy in \cite{Shannon46} as: \textquotedblleft No matter how much
material is intercepted, there is not a unique solution but many of
comparable probability.\textquotedblright\ In this work, we give a formal
definition of ideal secrecy.

\begin{definition}
\label{Def_IS copy(1)}A secrecy system is \emph{ideal} when%
\begin{equation}
\Pr ({ \mathbf{m}_{i}}|\left\{ \mathbf{y}_{i}\right\}
_{1}^{K})=\Pr \left\{ \mathbf{m}_{i}\mathbf{|y}_{i}\right\} =1/L_{i}\text{, for all }i%
\text{.}  \label{SE0}
\end{equation}
\end{definition}

\begin{remark}
In terms of Eve's equivocation, using the entropy chain rule, ideal secrecy
is achieved when%
\begin{equation}
\setlength{\abovedisplayskip}{3pt}\setlength{\belowdisplayskip}{3pt}%
H(\left\{ \mathbf{m}_{i}\right\} _{1}^{K}|\left\{ \mathbf{y}_{i}\right\}
_{1}^{K})=\sum_{i=1}^{K}H(\mathbf{m}_{i}\mathbf{|y}_{i})=\sum_{i=1}^{K}\log
L_{i}\text{,}  \label{SE1}
\end{equation}%
where $L_{i}\geq 2$ for all $i$. This condition will be used as our design
criterion for ideal secrecy.
\end{remark}

\subsection{Lattice Preliminaries}

To describe our scheme, it is convenient to introduce some lattice
preliminaries. An $n$-dimensional \emph{complex lattice} $\Lambda _{\mathbb{C%
}}$ in a complex space $\mathbb{C}^{m}$ ($n\leq m$) is the discrete set
defined by:%
\begin{equation*}
\Lambda _{\mathbb{C}}=\left\{ \mathbf{Bu}\text{: }\mathbf{u\in }\text{ }%
\mathbb{Z}\left[ i\right] ^{n}\right\} \text{,}
\end{equation*}%
where the \emph{basis} matrix $\mathbf{B=}\left[ \mathbf{b}_{1}\cdots
\mathbf{b}_{n}\right] $ has linearly independent columns.

$\Lambda _{\mathbb{C}}$ can also be easily represented as $2n$-dimensional
\emph{real} lattice $\Lambda _{\mathbb{R}}$ \cite{BK:Conway93}. In what
follows, we introduce some lattice parameters of $\Lambda _{\mathbb{C}}$,
which have a corresponding value for $\Lambda _{\mathbb{R}}$.

The \emph{Voronoi region} of $\Lambda _{\mathbb{C}}$, defined by%
\begin{equation*}
\mathcal{V}_{i}\left( \Lambda _{\mathbb{C}}\right) =\left\{ \mathbf{y}\in
\mathbb{C}^{m}\text{: }\Vert \mathbf{y}-\mathbf{x}_{i}\Vert \leq \Vert
\mathbf{y}-\mathbf{x}_{j}\Vert ,\forall \text{ }\mathbf{x}_{i}\neq \mathbf{x}%
_{j}\right\} ,
\end{equation*}%
gives the nearest neighbor decoding region of lattice point $\mathbf{x}_{i}$%
. The volume of any $\mathcal{V}_{i}\left( \Lambda _{\mathbb{C}}\right) $,
defined as vol$(\Lambda _{\mathbb{C}})\triangleq |\det (\mathbf{B}^{H}%
\mathbf{B})|$, is equivalent to the volume of the corresponding real
lattice. The \emph{effective radius} of $\Lambda _{\mathbb{C}}$, denoted by $%
r_{\text{eff}}(\Lambda _{\mathbb{C}})$, is the radius of a sphere$\,$of
volume vol$(\Lambda _{\mathbb{C}})$ \cite{Zamir08}. For large $n$, it is
approximately%
\begin{equation}
r_{\text{eff}}(\Lambda _{\mathbb{C}})\approx \sqrt{n/(\pi e)}\text{vol}%
(\Lambda _{\mathbb{C}})^{\frac{1}{2n}}.  \label{r_eff}
\end{equation}

\section{Unshared Secret Key Cryptosystem With Finite Constellation Inputs}

In this section, we show that the idea of USK can be applied to practical
systems using finite constellation inputs. We define the concept of secrecy
outage and a secrecy outage probability. We will later show how such
probability can be made arbitrarily small by considering larger
constellation size and jamming power.

\subsection{Encryption}

We consider a sequence of $K$ mutually independent messages $\left\{ \mathbf{%
m}_{i}\right\} _{1}^{K}$, where each one contains $n$ mutually independent
information bits. For each $\mathbf{m}$, Alice maps the corresponding $n$
bits to $N_{\text{A}}$ elements of $\mathbf{{{\mathbf{u}}}}$ for one channel
use. Each elements of $\mathbf{{{\mathbf{u}}}}$ is uniformly selected from a
finite lattice constellation $\mathcal{Q}{^{N_{\text{A}}}}$, where $\Re (%
\mathcal{Q}\mathbf{)}=\Im (\mathcal{Q}\mathbf{)}=\{0,1,...,\sqrt{M}-1\}$.
Consequently, we have%
\begin{equation}
\setlength{\abovedisplayskip}{3pt}\setlength{\belowdisplayskip}{3pt}n=N_{%
\text{A}}\log _{2}M\text{.}
\end{equation}%
Let $\mathbf{{{\mathbf{u}}}}_{i}$ be the transmitted vector corresponding to
message $\mathbf{m}_{i}$.

Across the $K$ channel uses, we apply a sequence of secret keys $\left\{
\mathbf{\hat{v}}_{i}\right\} _{1}^{K}$ to protect $\left\{ \mathbf{u}%
_{i}\right\} _{1}^{K}$. We consider secure transmissions in a fast fading
MIMO wiretap channel, i.e., all the channels $\mathbf{H}_{i}$, $\mathbf{G}%
_{i}$, $\mathbf{\hat{H}}_{\text{JB},i}$, $\mathbf{\hat{H}}_{\text{JE,}i}$
are assumed to be mutually independent and change independently for every
channel use.

We assume that $\left\{ \mathbf{\hat{v}}_{i}\right\} _{1}^{K}$ and $\left\{
\mathbf{u}_{i}\right\} _{1}^{K}$ are mutually independent. Using (\ref{SE0}%
), we only need to demonstrate the encryption process for one transmitted
vector $\mathbf{{{\mathbf{u}}}}_{i}$, corresponding to a message $\mathbf{m}%
_{i}$. We first interpret the signal model (\ref{WC_S}) as an encryption
algorithm:%
\begin{equation}
\setlength{\abovedisplayskip}{3pt}\setlength{\belowdisplayskip}{3pt}\mathbf{y%
}_{i}=\mathbf{G}_{i}\mathbf{u}_{i}+\mathbf{\hat{H}}_{\text{JE},i}\mathbf{%
\hat{Z}}_{i}\mathbf{\hat{v}}_{i}\text{.}  \label{N14}
\end{equation}

In detail, for the $i^{\text{th}}$ channel use, the jammers randomly and
independently (without any predefined distribution) choose a OTP $\mathbf{%
\hat{v}}_{i}$, from a ball of radius $\sqrt{P_{\text{J}}}$:%
\begin{equation}
\setlength{\abovedisplayskip}{3pt} \setlength{\belowdisplayskip}{3pt}
S\triangleq \{ \mathbf{\hat{v}}\in \mathbb{C}^{N_{\text{J}}-N\cdot N_{\text{B%
}}}\text{: }||\mathbf{\hat{v}||}^{2}\leq P_{\text{J}}\} \text{,}  \label{set}
\end{equation}%
and encrypts $\mathbf{{{\mathbf{u}}}}_{i}$ to $\mathbf{y}_{i}$ in (\ref{N14}%
) using $\mathbf{\hat{v}}_{i}$, such that $\mathbf{G}_{i}\mathbf{\mathbf{{{%
\mathbf{u}}}}}_{i}$ is the $k_{i}^{\text{th}}$ closest lattice point to $%
\mathbf{y}_{i}$, within the finite lattice%
\begin{equation}
\setlength{\abovedisplayskip}{3pt} \setlength{\belowdisplayskip}{3pt}
\Lambda _{\text{F},i}\triangleq \{\mathbf{G}_{i}\mathbf{\mathbf{{{\mathbf{u}}%
}}},{\mathbf{{\mathbf{u}}}}\in \mathcal{Q}{^{N_{\text{A}}}\}} \text{.}
\label{FN_LT}
\end{equation}%
The value of $k_{i}$ ranges from $1$ to $L_{i}$, where%
\begin{equation}
\setlength{\abovedisplayskip}{3pt} \setlength{\belowdisplayskip}{3pt}
L_{i}=|S_{R_{\max ,i}}\cap \Lambda _{\text{F},i}|\text{.}  \label{LII}
\end{equation}%
and $S_{R_{\max ,i}}$ is a sphere centered at $\mathbf{y}_{i}$ with radius:%
\begin{equation}
\setlength{\abovedisplayskip}{3pt} \setlength{\belowdisplayskip}{3pt}
R_{\max ,i}(P_{\text{J}})\triangleq \max_{||\mathbf{\hat{v}}_{i}\mathbf{||}%
^{2}\leq P_{\text{J}}}\left\Vert \mathbf{\hat{H}}_{\text{JE,}i}\mathbf{\hat{Z%
}}_{i}\mathbf{\hat{v}}_{i}\right\Vert =\sqrt{\lambda _{\max ,i}P_{\text{J}}}%
\text{.}  \label{R_maxi}
\end{equation}%
where $\lambda _{\max ,i}$ is the largest eigenvalue of $(\mathbf{\hat{H}}_{%
\text{JE,}i}\mathbf{\hat{Z}}_{i})^{H}(\mathbf{\hat{H}}_{\text{JE,}i}\mathbf{%
\hat{Z}}_{i})$.

As shown in Fig.~1, the full and empty dots correspond to the infinite
lattice%
\begin{equation}
\setlength{\abovedisplayskip}{3pt} \setlength{\belowdisplayskip}{3pt}
\Lambda _{\mathbb{C},i}\triangleq \{\mathbf{G}_{i}\mathbf{\mathbf{{{\mathbf{u%
}}}}},\mathbf{u}\in {\mathbb{Z}\left[ i\right] ^{N_{\text{A}}}\}}\text{.}
\label{Inf_L}
\end{equation}%
The finite subset of $\Lambda _{\mathbb{C},i}$, $\Lambda _{\text{F},i}$, is
demonstrated by the full dots. Consequently, $L_{i}$ represents the number
of full dots within the sphere $S_{R_{\max ,i}}$.

The security problem lies in how much Eve knows about $k_{i}$. Since we
assume that the realizations of $\mathbf{G}_{i},\mathbf{\hat{H}}_{\text{JE,}%
i},\mathbf{\hat{Z}}_{i}$ are known to Eve, $k_{i}$ is a function of $\mathbf{%
\hat{v}}_{i}$. Since $\mathbf{\hat{v}}_{i}$ is randomly and independently
selected by the jammers and is never shared with anyone, Eve can neither
know its realization nor its distribution. Thus, given $\mathbf{y}_{i}$, Eve
is not able to estimate the distribution of the index $k_{i}$, or more
specifically, she only knows that $\mathbf{G}_{i}\mathbf{\mathbf{{{\mathbf{u}%
}}}}_{i}\in S_{R_{\max ,i}}\cap \Lambda _{\text{F},i}$.

\begin{remark}
The index $k_{i}$ can be interpreted as the \emph{effective} one-time pad
secret key, whose randomness comes from $\mathbf{\hat{v}}_{i}$. The
effective key space size is $L_{i}$.
\end{remark}

\begin{figure}[tbp]
\centering\includegraphics[scale=0.75]{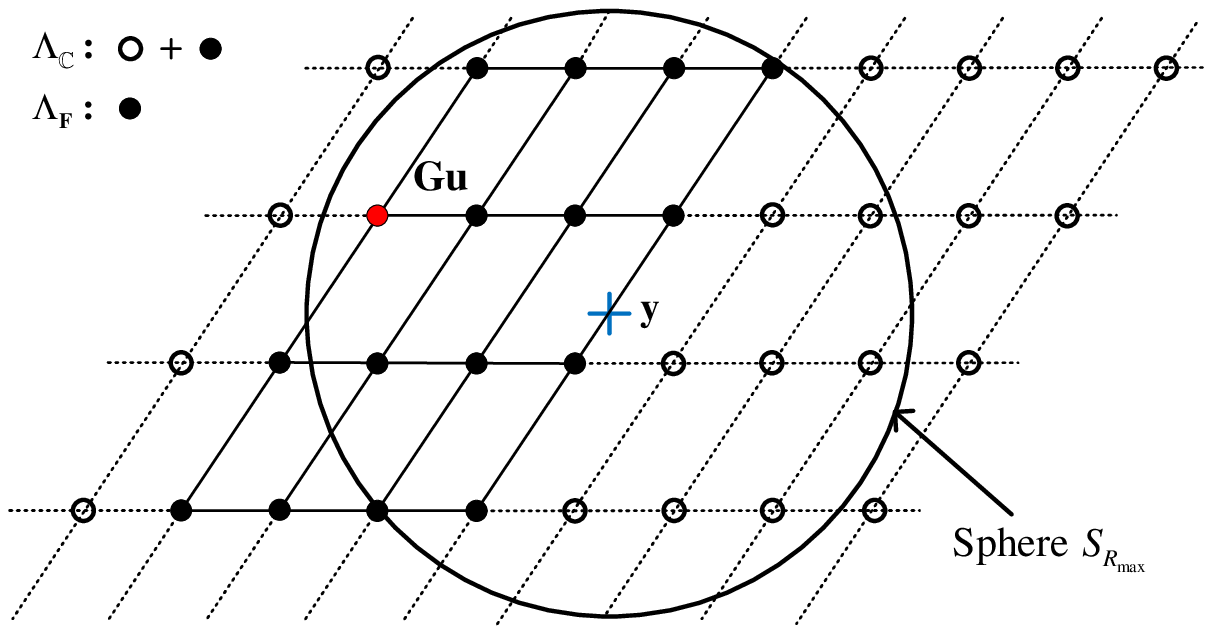}\vspace{-3mm}
\caption{ The USK cryptosystem with finite
constellations.}\vspace{-3 mm}
\end{figure}

\subsection{Analyzing Eve's Equivocation}
We then show that for each $\mathbf{u}_{i}$, Eve cannot obtain a
unique solution but $L_{i}$ indistinguishable candidates. The
posterior probability
that Eve obtains $\mathbf{{{\mathbf{u}}}}_{i}$, or equivalently, finds $%
k_{i} $, from the cryptogram $\mathbf{y}_{i}$, is%
\begin{equation}
\setlength{\abovedisplayskip}{3pt}\setlength{\belowdisplayskip}{3pt}\Pr
\left\{ \mathbf{{{\mathbf{u}}}}_{i}\mathbf{|y}_{i}\right\} =\Pr \left\{ k_{i}%
\mathbf{|y}_{i}\right\} =\Pr \left\{ \mathbf{{{\mathbf{u}}}}_{i}\mathbf{|{{%
\mathbf{u}}}}_{i}\in \mathcal{U}_{i}\right\} \text{,}  \label{PPP_FI}
\end{equation}%
where%
\begin{equation}
\setlength{\abovedisplayskip}{3pt}\setlength{\belowdisplayskip}{3pt}\mathcal{%
U}_{i}=\left\{ \mathbf{u}^{\prime }\text{: }\mathbf{G}_{i}\mathbf{\mathbf{{{%
\mathbf{u}}}}}^{\prime }\in S_{R_{\max ,i}}\cap \Lambda _{\text{F}%
,i}\right\} \text{.}  \label{UF}
\end{equation}

Due to the use of uniform constellation $\mathcal{Q}^{N_{\text{A}}}$,
according to Bayes' theorem, we have%
\begin{equation}
\setlength{\abovedisplayskip}{3pt} \setlength{\belowdisplayskip}{3pt} \Pr
\left\{ \mathbf{{{\mathbf{u}}}}_{i}\mathbf{|{{\mathbf{u}}}}_{i}\in \mathcal{U%
}_{i}\right\} =1/L_{i}\text{.}  \label{Post_FI}
\end{equation}

To recover the message $\mathbf{m}_{i}$, Eve has to recover the vector $\mathbf{{%
{\mathbf{u}}}}_{i}$, or equivalently, find $k_{i}$. Therefore, Eve's
equivocation is given by%
\begin{equation}
H(\mathbf{m}_{i}\mathbf{|y}_{i})=H(k_{i}\mathbf{|y}_{i})=H(\mathbf{{{\mathbf{%
u}}}}_{i}|\mathbf{y}_{i})\text{.}  \label{uncer_FI1}
\end{equation}%
Moreover, since $\mathbf{u}_{i}$ is independent of $\mathbf{u}_{j}$ and $%
\mathbf{y}_{j}$, we have%
\begin{equation}
H(\left\{ \mathbf{{{\mathbf{u}}}}_{i}\right\} _{1}^{K}|\left\{ \mathbf{y}%
_{i}\right\} _{1}^{K})=\sum_{i=1}^{K}H(\mathbf{u}_{i}\mathbf{|y}%
_{i})=\sum_{i=1}^{K}\log L_{i}\text{.}  \label{uncer_FI}
\end{equation}

\begin{remark}
From (\ref{SE1}) and (\ref{uncer_FI}), ideal secrecy is achieved if%
\begin{equation}
L_{i}\geq 2,~\text{for all}~i\text{.}  \label{L_CON}
\end{equation}
\end{remark}

\subsection{Ideal Secrecy Outage}

We then study how to satisfy the condition in (\ref{L_CON}). Note that the
values in $\left\{ L_{i}\right\} _{1}^{K}$ are known to Eve, but not Alice.
From Alice's perspective, according to (\ref{LII}) and (\ref{UF}), $L_{i}$
is a function of $\mathbf{G}_{i}$ $\mathbf{\hat{H}}_{\text{JE,}i},\mathbf{%
\hat{Z}}_{i},\mathbf{\hat{v}}_{i}$, thus a random variable. Although Alice
cannot know the exact values in $\left\{ L_{i}\right\} _{1}^{K}$, she may be
able to evaluate the joint probability $\Pr \left\{ L_{1}\geq d\text{, ..., }%
L_{K}\geq d\right\} $, for any $2\leq d\leq M^{N_{\text{A}}}$.

We refer to the event%
\begin{equation}
L_{i}<d\text{,}
\end{equation}%
as the \emph{ideal secrecy outage}. We refer to the probability
\begin{equation}
P_{\text{out}}(d,K)=1-\Pr \left\{ L_{1}\geq d\text{, ..., }L_{K}\geq
d\right\} \text{,}  \label{PF_out}
\end{equation}%
as the secrecy outage probability. From (\ref{L_CON}) and (\ref{PF_out}), if
$P_{\text{out}}(d,K)\rightarrow 0$, then%
\begin{equation}
L_{i}\geq d,~\text{for all}~i\text{.}
\end{equation}

In the next section, we will show that Alice can ensure $P_{\text{out}%
}(d,K)\rightarrow 0$ by increasing the jamming power $P_{\text{J}}$
and constellation size $M$.

\section{The Security of USK}

In this section, we show that the USK with the finite constellation $%
\mathcal{Q}^{N_{\text{A}}}$ provides Shannon's ideal secrecy with an
arbitrarily small outage.

\subsection{A Useful Lemma}

We define%
\begin{equation}
\Theta (P_{\text{J}})\triangleq \frac{2R_{\max }(P_{\text{J}})}{\sqrt{M}r_{%
\text{eff}}(\Lambda _{\mathbb{C}})}\text{,}  \label{RRR}
\end{equation}%
where $\Lambda _{\mathbb{C}}$ is given in (\ref{Inf_L}), $r_{\text{eff}%
}(\Lambda _{\mathbb{C}})$ is given in (\ref{r_eff}), and $R_{\max }(P_{\text{%
J}})$ is given in (\ref{R_maxi}).

From Alice's perspective, $\Theta (P_{\text{J}})$ is a function of $\mathbf{%
G,\hat{H}}_{\text{JE,}i},\mathbf{\hat{Z}}_{i}$, thus is a random variable.
To prove the main theorem, we first introduce the following lemma.

\begin{lemma}
\label{Th4}%
\begin{align}
& \Pr \left\{ \Theta (P_{\text{J}})<x\right\}  \notag \\
& >\prod\limits_{j=1}^{N_{\text{A}}}B_{\substack{ N_{\text{E}}N_{\text{J}},
\\ N_{\text{E}}-j+1}}\left( \frac{N_{\text{E}}N_{\text{J}}g(x,j)}{N_{\text{E}%
}N_{\text{J}}g(x,j)+N_{\text{E}}-j+1}\right) \text{,}  \label{P_t}
\end{align}%
where%
\begin{equation}
g(x,j)=\frac{x^{2}MN_{\text{A}}(N_{\text{E}}-j+1)}{4\pi eP_{\text{J}}N_{%
\text{E}}N_{\text{J}}(N_{\text{J}}-N\cdot N_{\text{B}})},  \label{g_x}
\end{equation}%
and $B_{a\text{,}b}(x)$ is the \emph{regularized incomplete beta} function
\cite{Imbeta}:%
\begin{equation}
B_{a\text{,}b}(x)\triangleq \sum_{j=a}^{a+b-1}\left(
\begin{array}{c}
a+b-1 \\
j%
\end{array}%
\right) x^{j}(1-x)^{a+b-1-j}\text{.}  \label{Beta}
\end{equation}
\end{lemma}

\begin{proof}
See Appendix A.
\end{proof}

\subsection{Ideal Secrecy Outage Probability}

An upper bound on $P_{\text{out}}(d,K)$ in (\ref{PF_out}) can be derived
using Lemma \ref{Th4}.

\begin{theo}
\label{Th3}Given $\varepsilon <1$, $d\geq 2$, $M\geq \varepsilon
^{-3-2/N_{\min }}\kappa (d)^{2}$, and $P_{\text{J}}= \varepsilon
^{-2/N_{\min }}\kappa (d)^{2}/\Phi ^{2N_{\text{A}}/N_{\text{E}}}$, then%
\begin{equation}
\setlength{\abovedisplayskip}{3pt}\setlength{\belowdisplayskip}{3pt}P_{\text{%
out}}(d,K)<O(\varepsilon )\text{,}
\end{equation}%
where%
\begin{equation}
N_{\min }\triangleq \min \left\{ N_{\text{E}}-N_{\text{A}}+1\text{, }N_{%
\text{A}}\right\} \text{,}  \label{n_min}
\end{equation}%
\begin{equation}
\kappa (d)\triangleq d^{1/(2N_{\text{E}})}/\sqrt{\pi }\text{,}  \label{kappa}
\end{equation}%
\begin{equation}
\Phi \triangleq \left[ \frac{(N_{\text{E}}-N_{\text{A}})!}{N_{\text{E}}!}%
\right] ^{\frac{1}{2N_{\text{A}}}}\text{,}  \label{Pi_SVD}
\end{equation}%
i.e., ideal secrecy is achieved with probability $1-O(\varepsilon )$.
\end{theo}

\begin{proof}
See Appendix B.
\end{proof}

Theorem \ref{Th3} shows that for finite constellation $\mathcal{Q}^{N_{\text{%
A}}}$, the ideal secrecy outage can be made arbitrarily small. Given
a target pair $\left\{ \varepsilon ,d\right\} $, we can easily
compute the required values of $P_{\text{J}}$ and $M$ to realize the
USK cryptosystem.

\begin{example}
Fig.~2 examines the value of $P_{\text{out}}(2,1)$ as a function of $%
\varepsilon $. We choose $P_{\text{J}}$ and $M$ according to%
\begin{equation*}
\setlength{\abovedisplayskip}{3pt}
\setlength{\belowdisplayskip}{3pt}
P_{\text{J}}=\varepsilon ^{-2/N_{\min }}\kappa (d)^{2}/\Phi ^{2N_{\text{A}%
}/N_{\text{E}}}\text{ and }M\geq \varepsilon ^{-3-2/N_{\min }}\kappa (d)^{2}%
\text{.}
\end{equation*}%
We observe that $P_{\text{out}}(2,1)=1.6\times 10^{-4}$ when $P_{\text{J}%
}=3.5926$ and $M=256$. This simulation confirms that the secrecy outage for
the finite constellation $\mathcal{Q}^{N_{\text{A}}}$ can be made
arbitrarily small by increasing $P_{\text{J}}$ and $M$.
\end{example}

\begin{remark}
Using the proof of Theorem \ref{Th3}, we can show%
\begin{equation}
\Pr \left\{ L_{1}=1\text{, ..., }L_{K}=1\right\} <O(\varepsilon ^{K})\text{.}
\end{equation}%

In order to enhance security, Alice can scramble the $nK$ message
bits in $\{ {\bf m'}_i\}_1^K$ by using an invertible binary
$nK\times nK$ matrix ${\bf S}$ to produce the sequence
 \begin{equation} \label{scrambler}
\{ {\bf m}_i\}_1^K = \{ {\bf m'}_i\}_1^K {\bf S}\text{.}
\end{equation}
Only with a probability of the order of $\varepsilon^K$ given in
(37) Eve would be able to uniquely recover  $\{ {\bf m'}_i\}_1^K$ by
inverting ${\bf S}$. In all other cases, occurring with probability
$\varepsilon^{K'}$, where Eve recovers only $K'<K$ sub-blocks ${\bf
m'}_i $,  the system in (\ref{scrambler}) will have $2^{n(K-K')}$
different solutions. This enhances the overall ambiguity of Eve's
when guessing the entire message  $\{{\bf m'}_i\}_1^K$.
\end{remark}

\begin{figure}[tbp]
\centering\includegraphics[scale=0.4]{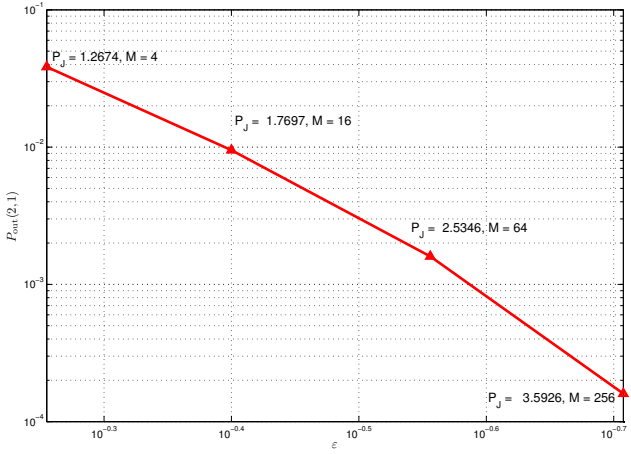}\vspace{-3mm}
\caption{ $P_{\text{out}}(2,1)$ vs. $\protect\varepsilon $ with $N_{\text{A}%
}=N_{\text{B}}=2$, $N_{\text{J,1}}=N_{\text{J,2}}=3$, and
$N_{\text{E}}=4$.}\vspace{-6mm}
\end{figure}

\section{Conclusions}

In this work, we showed how to construct a practical unshared secret key
(USK) cryptosystem\ using finite constellation inputs and some helpers. The
new USK scheme is specially designed for the scenario $N_{\text{E}}\geq N_{%
\text{A}}$, where the original USK scheme is not valid. We have shown that
Shannon's ideal secrecy can be obtained with an arbitrarily small outage
probability, by simply increasing the constellation size and jamming power.
Our results provide new ideas for the innovations and combinations of
cryptography and physical layer security. Future work will generalize USK to
relaying networks.

\section*{Appendix}

\subsection{Proof of Lemma 1}

Recalling that%
\begin{equation}
R_{\max }(P_{\text{J}})=\max_{||\mathbf{\hat{v}||}^{2}\leq P_{\text{J}%
}}\left\Vert \mathbf{\hat{H}}_{\text{JE}}\mathbf{\hat{Z}\hat{v}}\right\Vert
\text{,}  \label{aabb}
\end{equation}%
\begin{equation}
r_{\text{eff}}(\Lambda _{\mathbb{C}})=\sqrt{N_{\text{A}}/(\pi e)}|\det (%
\mathbf{G}^{H}\mathbf{G)|}^{\frac{1}{2N_{\text{A}}}}\text{.}
\end{equation}

From (\ref{R_maxi}), applying Cauchy--Schwarz inequality,%
\begin{align}
R_{\max }^{2}(P_{\text{J}})& \leq P_{\text{J}}\left\Vert \mathbf{\hat{H}}_{%
\text{JE}}\mathbf{Z}\right\Vert _{F}^{2}\leq P_{\text{J}}\left\Vert \mathbf{%
\hat{H}}_{\text{JE}}\right\Vert _{F}^{2}\left\Vert \mathbf{Z}\right\Vert
_{F}^{2}  \notag \\
& =P_{\text{J}}(N_{\text{J}}-N\cdot N_{\text{B}})\left\Vert \mathbf{\hat{H}}%
_{\text{JE}}\right\Vert _{F}^{2}\text{.}
\end{align}

From Alice perspective, $\mathbf{\hat{H}}_{\text{JE}}$ is a complex Gaussian
random matrix with i.i.d. components. Thus, $\left\Vert \mathbf{\hat{H}}_{%
\text{JE}}\right\Vert _{F}^{2}$ can be expressed in terms of a Chi-squared
random variable:%
\begin{equation}
\left\Vert \mathbf{\hat{H}}_{\text{JE}}\right\Vert _{F}^{2}=\frac{1}{2}%
\mathcal{X}^{2}\left( 2N_{\text{E}}N_{\text{J}}\right) \text{.}  \label{RM1}
\end{equation}

According to \cite{Verdu04}, $r_{\text{eff}}(\Lambda _{\mathbb{C}})$ can be
expressed in terms of $N_{\text{A}}$ independent Chi-squared variables:%
\begin{equation}
r_{\text{eff}}(\Lambda _{\mathbb{C}})=\sqrt{N_{\text{A}}/(\pi e)}\left(
\prod_{j=1}^{N_{\text{A}}}\frac{1}{2}\mathcal{X}^{2}\left( 2(N_{\text{E}%
}-j+1)\right) \right) ^{\frac{1}{2N_{\text{A}}}}\text{.}  \label{RE2}
\end{equation}

Moreover, since $\mathbf{G},\mathbf{\hat{H}}_{\text{JE}}\mathbf{,\hat{Z}}$
are mutually independent, $R_{\max }(P_{\text{J}})$ and $r_{\text{eff}%
}(\Lambda _{\mathbb{C}})$ are independent.

Then, we have
\begin{align}
& \Pr \left\{ \dfrac{2R_{\max }(P_{\text{J}})}{\sqrt{M}r_{\text{eff}%
}(\Lambda _{\mathbb{C}})}<x\right\}  \notag \\
& \geq \Pr \left\{ \frac{P_{\text{J}}\left\Vert \mathbf{\hat{H}}_{\text{JE}%
}\right\Vert _{F}^{2}}{r_{\text{eff}}(\Lambda _{\mathbb{C}})^{2}}<\frac{%
x^{2}M}{4(N_{\text{J}}-N\cdot N_{\text{B}})}\right\}  \notag \\
& \overset{(a)}{\geq }\Pr \left\{ \dfrac{\mathcal{X}^{2}\left( 2N_{\text{E}%
}N_{\text{J}}\right) }{\dfrac{N_{\text{A}}}{\sum_{j=1}^{N_{\text{A}}}\dfrac{1%
}{\mathcal{X}^{2}\left( 2(N_{\text{E}}-j+1)\right) }}}<\frac{x^{2}MN_{\text{A%
}}}{4\pi eP_{\text{J}}(N_{\text{J}}-N\cdot N_{\text{B}})}\right\}  \notag \\
& =\Pr \left\{ \sum_{j=1}^{N_{\text{A}}}\frac{\mathcal{X}^{2}\left( 2N_{%
\text{E}}N_{\text{J}}\right) }{\mathcal{X}^{2}\left( 2(N_{\text{E}%
}-j+1)\right) }<\frac{x^{2}MN_{\text{A}}^{2}}{4\pi eP_{\text{J}}(N_{\text{J}%
}-N\cdot N_{\text{B}})}\right\}  \notag \\
& \overset{(b)}{>}\prod\limits_{j=1}^{N_{\text{A}}}\Pr \left\{ \frac{%
\mathcal{X}^{2}\left( 2N_{\text{E}}N_{\text{J}}\right) }{\mathcal{X}%
^{2}\left( 2(N_{\text{E}}-j+1)\right) }\leq \frac{x^{2}MN_{\text{A}}}{4\pi
eP_{\text{J}}(N_{\text{J}}-N\cdot N_{\text{B}})}\right\}  \notag \\
& =\prod\limits_{j=1}^{N_{\text{A}}}\Pr \left\{ \mathcal{F}(2N_{\text{E}}N_{%
\text{J}}\text{, }2(N_{\text{E}}-j+1))\leq g(x,j)\right\} \text{,}
\label{xxx}
\end{align}%
where $g(x,j)$ is given in (\ref{g_x}), and $\mathcal{F}(k_{1}$, $k_{2})$
represents an $\mathcal{F}$-distributed random variable with $k_{1}$ and $%
k_{2}$ degrees of freedom. ($a$) holds due to the inequality of geometric
and harmonic means. ($b$) holds by induction on the fact that if the
non-negative random variables $A_{i}$, $1\leq i\leq N$, are mutually
independent, given a constant $C>0$,%
\begin{align}
& \Pr \left\{ \sum_{i=1}^{N}A_{i}<C\right\} >\Pr \left\{ A_{1}\leq
C/N;\sum_{i=2}^{N}A_{i}\leq C(N-1)/N\right\}  \notag \\
& =\Pr \left\{ A_{1}\leq C/N\right\} \Pr \left\{ \sum_{i=2}^{N}A_{i}\leq
C(N-1)/N\right\} \text{.}
\end{align}

Since the cdf of $\mathcal{F}(k_{1}$, $k_{2})$ can be expressed using the
regularized incomplete beta function \cite{Imbeta}, the final expression of (%
\ref{xxx}) is given in (\ref{P_t}).

\QEDA\vspace{-3mm}

\subsection{Proof of Theorem 1}

From Alice's perspective, $L_{i}$ is a function of $\mathbf{G}_{i}$, $%
\mathbf{\hat{H}}_{\text{JE},i}$, $\mathbf{,\hat{Z}}_{i}$, and $\mathbf{\hat{v%
}}_{i}$. Since $\left\{ \mathbf{G}_{i}\right\} _{1}^{K},\left\{ \mathbf{\hat{%
H}}_{\text{JE},i}\right\} _{1}^{K},\left\{ \mathbf{\hat{Z}}_{i}\right\}
_{1}^{K},\left\{ \mathbf{\hat{v}}_{i}\right\} _{1}^{K}$are mutually
independent, $\left\{ L_{i}\right\} _{1}^{K}$ are mutually independent. From
(\ref{PF_out}),
\begin{equation}
\setlength{\abovedisplayskip}{3pt}\setlength{\belowdisplayskip}{3pt}P_{\text{%
out}}(d,K)=1-\prod\limits_{i=1}^{K}\Pr \left\{ L_{i}\geq d\right\} \text{.}
\label{122}
\end{equation}

We then evaluate $\Pr \left\{ L_{i}<d\right\} $. For simplicity, we remove
the index $i$. We define%
\begin{equation}
\setlength{\abovedisplayskip}{3pt}\setlength{\belowdisplayskip}{3pt}%
D=|S_{R_{\max }}\cap \Lambda _{\mathbb{C}}|\text{.}  \label{Di}
\end{equation}%
According to \cite[Th. 2]{Liu147}, with $P_{\text{J}}= \varepsilon
^{-2/N_{\min }}\kappa (d)^{2}/\Phi ^{2N_{\text{A}}/N_{\text{E}}}$,
the
jammers can ensure%
\begin{equation}
\setlength{\abovedisplayskip}{3pt}\setlength{\belowdisplayskip}{3pt}\Pr
(D<d)<O(\varepsilon )\text{,}
\end{equation}%
where $N_{\min }$ is given in (\ref{n_min}), $\kappa (d)$ is given in (\ref%
{kappa}), and $\Phi $ is given in (\ref{Pi_SVD}). We can upper bound $\Pr
\left\{ L<d\right\} $ by%
\begin{align}
& \Pr \left\{ L<d\right\} =\Pr \{L<d|D\geq d\}\Pr \{D\geq d\}  \notag \\
& +\Pr \{L<d|D<d\}\Pr \{D<d\}  \notag \\
& \leq \Pr \left\{ L<D|D\geq d\right\} \Pr \{D\geq d\}+O(\varepsilon )
\notag \\
& \leq \Pr \left\{ L<D\right\} +O(\varepsilon )\text{.}  \label{d_D}
\end{align}%
We then evaluate $\Pr \left\{ L<D\right\} $.%
\begin{align}
& \Pr \left\{ L<D\right\} =\Pr \{L<D|\Theta (P_{\text{J}})<\varepsilon \}\Pr
\{\Theta (P_{\text{J}})<\varepsilon \}  \notag \\
& +\Pr \{L<D|\Theta (P_{\text{J}})\geq \varepsilon \}\Pr \{\Theta (P_{\text{J%
}})\geq \varepsilon \}  \notag \\
& \leq \Pr \{L<D|\Theta (P_{\text{J}})<\varepsilon \}+\Pr \{\Theta (P_{\text{%
J}})\geq \varepsilon \}\text{,}  \label{L_E}
\end{align}%
where $\Theta (P_{\text{J}})$ is given in (\ref{RRR}).

We then evaluate the two terms in (\ref{L_E}), separately.

\emph{1)} $\Pr \left\{ L<D|\Theta (P_{\text{J}})<\varepsilon \right\} $\emph{%
:} Recalling that%
\begin{equation}
\setlength{\abovedisplayskip}{3pt} \setlength{\belowdisplayskip}{3pt}
\mathbf{y}=\mathbf{G{{\mathbf{u+}}}\hat{H}}_{\text{JE}}\mathbf{\hat{Z}\hat{v}%
}\text{ and }\Lambda _{\text{F}}=\{\mathbf{G\mathbf{{{\mathbf{u}}}}},{%
\mathbf{{\mathbf{u}}}}\in \mathcal{Q}{^{N_{\text{A}}}\}}\text{.}
\end{equation}

Since $L=|S_{R_{\max }}\cap \Lambda _{\text{F}}|$, we begin by checking the
boundary of $\Lambda _{\text{F}}$. Let $\mathbf{O}$ be the center point of $%
\Lambda _{\text{F}}$. According to \cite{Ajtai02}, for the Gaussian random
lattice basis $\mathbf{G}$, the boundary of $\Lambda _{\text{F}}$ can be
approximated by a sphere $S_{\text{F,S}}$ centered at $\mathbf{O}$ with
radius $\sqrt{M}r_{\text{eff}}(\Lambda _{\mathbb{C}})$, where $r_{\text{eff}%
}(\Lambda _{\mathbb{C}})$ is given in (\ref{r_eff}).

Given $\Theta (P_{\text{J}})<\varepsilon $ and $\varepsilon <1$, we have $%
\sqrt{M}r_{\text{eff}}(\Lambda _{\mathbb{C}})>2R_{\max }(P_{\text{J}})$. We
define a concentric sphere $S_{\text{F,C}}$ with radius $\sqrt{M}r_{\text{eff%
}}(\Lambda _{\mathbb{C}})-2R_{\max }(P_{\text{J}})$, where $R_{\max }(P_{%
\text{J}})$ is given in (\ref{R_maxi}). We then check when $L=D$ given $%
\Theta (P_{\text{J}})<\varepsilon $.

If $\mathbf{G{{\mathbf{u}}}}\in S_{\text{F,C}}$, using triangle inequality,
we have%
\begin{align}
||\mathbf{y-O||} \leq \left\Vert \mathbf{G{{\mathbf{u}}}}-\mathbf{O}%
\right\Vert +\left\Vert \mathbf{\hat{H}}_{\text{JE}}\mathbf{\hat{Z}\hat{v}}%
\right\Vert \leq \sqrt{M}r_{\text{eff}}(\Lambda _{\mathbb{C}})-R_{\max }(P_{%
\text{J}})  \label{b1}
\end{align}%
We then check the locations of the $D$ elements in $S_{R_{\max }}\cap
\Lambda _{\mathbb{C}}$ (\ref{Di}), denoted by, $\mathbf{Gu}_{t}^{\prime }$, $%
1\leq t\leq D$. Note that%
\begin{equation}
\setlength{\abovedisplayskip}{3pt} \setlength{\belowdisplayskip}{3pt}
\left\Vert \mathbf{Gu}_{t}^{\prime }-\mathbf{{{\mathbf{y}}}}\right\Vert \leq
R_{\max }(P_{\text{J}})\text{.}  \label{b2}
\end{equation}%
From (\ref{b1}) and (\ref{b2}), using triangle inequality, for all $t$,%
\begin{equation}
\setlength{\abovedisplayskip}{3pt} \setlength{\belowdisplayskip}{3pt}
\left\Vert \mathbf{Gu}_{t}^{\prime }-\mathbf{O}\right\Vert \leq \left\Vert
\mathbf{y-O}\right\Vert +\left\Vert \mathbf{Gu}_{t}^{\prime }-\mathbf{{{%
\mathbf{y}}}}\right\Vert \leq \sqrt{M}r_{\text{eff}}(\Lambda _{\mathbb{C}})%
\text{.}
\end{equation}%
Therefore, $S_{R_{\max }}\cap \Lambda _{\mathbb{C}}\subset \Lambda _{\text{F}%
}$, i.e., $L=D$.

If $\mathbf{G\mathbf{\mathbf{V}}}_{1}\mathbf{{{\mathbf{u}}}}\notin S_{\text{%
F,C}}$, there is a probability that $L<D$. Therefore, we have%
\begin{equation}
\setlength{\abovedisplayskip}{3pt} \setlength{\belowdisplayskip}{3pt} \Pr
\left\{ L<D|\Theta (P_{\text{J}})<\varepsilon \right\} <\Pr \left\{ \mathbf{G%
{{\mathbf{u}}}}\notin S_{\text{F,C}}\right\} \text{.}  \label{EEE1}
\end{equation}%
Since $\mathbf{G{{\mathbf{u}}}}$ is uniformly distributed over $S_{\text{F,S}%
}$, we have%
\begin{equation}
\setlength{\abovedisplayskip}{3pt} \setlength{\belowdisplayskip}{3pt} \Pr
\left\{ \mathbf{G{{\mathbf{u}}}}\in S_{\text{F,C}}\right\} =\frac{\text{vol}%
(S_{\text{F,C}})}{\text{vol}(S_{\text{F,S}})}=\left( 1-\Theta (P_{\text{J}%
})\right) ^{2N_{\text{A}}}>\left( 1-\varepsilon \right) ^{2N_{\text{A}}}
\label{EEE2}
\end{equation}

Based on (\ref{EEE1}) and (\ref{EEE2}), we have%
\begin{equation}
\setlength{\abovedisplayskip}{3pt} \setlength{\belowdisplayskip}{3pt} \Pr
\left\{ L<D|\Theta (P_{\text{J}})<\varepsilon \right\} <1-\left(
1-\varepsilon \right) ^{2N_{\text{A}}}=O(\varepsilon )\text{.}  \label{pv1}
\end{equation}

\emph{2)} $\Pr \{\Theta (P_{\text{J}})\geq \varepsilon \}$\emph{:} Using
Lemma \ref{Th4} with $M\geq \varepsilon ^{-3-2/N_{\min }}\kappa (d)^{2}$, we
have%
\begin{align}
& \Pr \left\{ \Theta (P_{\text{J}})<\varepsilon \right\} \geq
\prod\limits_{j=1}^{N_{\text{A}}}B_{a,b(j)}\left( 1-\frac{b(j)}{%
ag(\varepsilon ,j)+b(j)}\right)  \notag \\
&\!\overset{(a)}{=}\prod\limits_{j=1}^{N_{\text{A}}}1-B_{b(j),a}\left( \frac{%
b(j)}{ag(\varepsilon ,j)+b(j)}\right) \overset{(b)}{=}\prod\limits_{j=1}^{N_{%
\text{A}}}\left( 1-O(\varepsilon ^{b(j)})\right)  \notag \\
&>(1-O(\varepsilon ^{\hat{N}})) ^{N_{\text{A}}}\text{,}
\end{align}%
where $\hat{N}=N_{\text{E}}-N_{\text{A}}+1$ and%
\begin{equation}
a=N_{\text{E}}N_{\text{J}}\text{ and }b(j)=N_{\text{E}}-j+1\text{.}
\end{equation}%
$(a)$ and $(b)$ hold due to the facts that%
\begin{equation}
B_{a,b(j)}(x)=1-B_{b(j),a}(1-x)\text{,}
\end{equation}%
\begin{equation}
B_{b(j),a}(x)=O(x^{b(j)})\text{, for }x\rightarrow 0\text{.}
\end{equation}

Consequently, we have%
\begin{equation}
\Pr \{\Theta (P_{\text{J}})\geq \varepsilon \}<1-\left( 1-O(\varepsilon ^{%
\hat{N}})\right) ^{N_{\text{A}}}=O(\varepsilon ^{\hat{N}})\text{.}
\label{pv2}
\end{equation}

By substituting (\ref{L_E}), (\ref{pv1}) and (\ref{pv2}) to (\ref{d_D}), we
have%
\begin{equation}
\Pr \left\{ L<d\right\} <O(\varepsilon )\text{.}  \label{tt1}
\end{equation}

From (\ref{122}) and (\ref{tt1}), if $M\geq\varepsilon
^{-3-2/N_{\min }}\kappa (d)^{2}$ and $P_{\text{J}}= \varepsilon
^{-2/N_{\min }}\kappa
(d)^{2}/\Phi ^{2N_{\text{A}}/N_{\text{E}}}$, we have%
\begin{equation}
P_{\text{out}}(d,K)<1-(1-O(\varepsilon ))^{K}=O(\varepsilon )\text{.}
\end{equation}%
\QEDA

\vspace{-3 mm}

\bibliographystyle{IEEEtran}
\bibliography{IEEEabrv,LIUBIB}

\end{document}